\documentclass[conference]{IEEEtran}                                             
\usepackage{amsmath}
\usepackage{graphics}
\usepackage{graphicx}
\usepackage{amssymb}
\usepackage{amsthm}
\usepackage[usenames,dvipsnames,svgnames,table]{xcolor}
\usepackage{tikz}
\usepackage{multirow}
\usepackage{subcaption}
\usepackage{cite}

\newcommand{\LyX}{L\kern-.1667em\lower.25em\hbox{Y}\kern-.125emX\spacefactor1000}

\title{The Economics of Competition and Cooperation Between MNOs and MVNOs}

\author{\IEEEauthorblockN{Mohammad Hassan Lotfi}
\IEEEauthorblockA{Department of Electrical and System Engineering\\
University of Pennsylvania\\
Philadelphia, PA, 19104 \\
Email: lotfm@seas.upenn.edu}
\and
\IEEEauthorblockN{Saswati Sarkar}
\IEEEauthorblockA{Department of Electrical and System Engineering\\
University of Pennsylvania\\
Philadelphia, PA, 19104 \\
Email: swati@seas.upenn.edu}
}

\begin{document}
\maketitle
\newtheorem{lemma}{Lemma}
\newtheorem{note}{Note}
\newtheorem{property}{Property}
\newtheorem{theorem}{Theorem}
\newtheorem{definition}{Definition}
\newtheorem{corollary}{Corollary}
\newtheorem{remark}{Remark}
\newtheorem{assumption}{Assumption}

\begin{abstract}
In this work, we consider the economics of the interaction between Mobile Virtual Network Operators (MVNOs) and Mobile Network Operators (MNOs). We investigate the incentives of an MNO for offering some of her resources to an MVNO instead of using the resources for her own. We formulate the problem as a sequential game. We consider a market with one MNO and one MVNO, and a continuum of undecided end-users. We assume that EUs have different preferences for the MNO and the MVNO. These preferences can be because of the differences in the service they are offering or the reluctance of  an EU to buy her plan from one of them.  We assume that the preferences also depend on the investment level the MNO and the MVNO.  We show that there exists a unique interior  SPNE, i.e. the SPNE by which both SPs receive a positive mass of EUs, and characterize it. We also consider a benchmark case in which the MNO and the MVNO do not cooperate, characterize the unique SPNE of this case, and compare the results of our model to the benchmark case to assess the incentive of the MNO to invest in her infrastructure and to offer it to the MVNO. 
\end{abstract}

\section{Introduction}

Traditionally, wireless services have been offered by Service Providers (SPs) that own the infrastructure they are operating on.  Nowadays SPs are divided into (i) Mobile Network Operators (MNOs) that own the infrastructure, and (ii) 
Mobile Virtual Network Operators (MVNOs)  that do not own the infrastructure they are operating on, and use the resources of one or more MNOs based on a business contract. MVNOs can distinguish their plans from MNOs by  bundling their service with other products,  offering different pricing plans  for End-Users (EUs), or building  a good reputation through a better customer service. In recent years, the number of MVNOs is rapidly growing. According to \cite{GSMA}, between June 2010 and June 2015, the number of MVNOs increased by 70 percent worldwide,  reaching 1,017 as of June 2015.  Even some MNOs developed their own MVNOs. An example of which is Cricket wireless which is owned by At\&t and offers a prepaid wireless service to EUs.

In this work, we consider the economics of the interaction between  MVNOs and MNOs.  We investigate the incentives of an MNO for offering some of her resources to an MVNO instead of using the resources for her own. Thus, we consider competition and cooperation between an MNO  and an MVNO. Note that it is not apriori known  how much an MNO is willing to invest on the infrastructure and how much an MVNO is willing to lease. More importantly, it is not apriori clear that under what conditions the MNO prefers to generate her revenue through the MVNO by leasing the resources to her, and therefore letting the MVNO to attract EUs.

Many works have considered the economics of resource/spectrum sharing and the subsequent profit sharing between SPs. Examples are \cite{cano2016cooperative,berry2013nature,le2009pricing,singh2012cooperative,korcak2012collusion
,banerjee2009voluntary,lotfi2016uncertain,sponsoring_journal}. In these works, authors model the environment using game theory and seek to provide intuitions for the pricing schemes, the split of EUs/benefits between SPs, and the investment level of different SPs. In this work, however, as mentioned before, we focus on the competition between the MNOs and MVNOs.

We consider a market with one MNO (e.g. At\&t) and  one MVNO (e.g. Google's Project Fi). We assume that the MVNO is active, i.e. has already some resources, and is  willing to lease additional resources from the MNO   in exchange of a fee. The MNO decides on the investment level. Subsequently, knowing this decision, the MVNO decides on the number of resources she wants to lease from the MNO. Then, simultaneously, both MNO and MNVO decide on their pricing strategies for EUs.   Finally, EUs choose one of the MNO or MVNO to buy the wireless plan from. We assume the per resource fee that the MVNO pays to the MNO to be fixed and discuss about the framework for determining this fee.

To model EUs, we consider a continuum of undecided EUs that decide which of the SPs they want to buy their wireless plan from. We assume that EUs have different preferences for each SP. These preferences can be because of different services that SPs offer. For example, the MVNO can bundle the wireless service with a free or cheap international call plan (through VoIP) to make her service more favorable for some EUs. Moreover, the preferences can be because of the reluctance of  an EU to buy her wireless plan from a particular SP (e.g. an  SP with an infamous customer service). We assume that the preferences also depend on the investment level of the SPs. In other words, the higher the investment level of an SP, the lower would be the reluctance of an EU for that SP. 

Different preferences of EUs for SPs enables the possibility of different outcomes for the market. For example, the MVNO who bundles her service with free international call may attract some of EUs even with a higher price than what MNO is offering. Thus,  instead of competing for EUs (by lowering her price), the MNO can lease some of her resources to this MVNO and receives her share of profit through the MVNO. 

We formulate the game as a sequential game, and seek the
Sub-game Perfect Nash Eequilibrium (SPNE) of the sequential
game using backward induction. We show that there exists a unique interior  SPNE, i.e. the SPNE by which both SPs receive a positive mass of EUs, and characterize it.  Moreover, we consider a benchmark case in which the MNO does not invest in her infrastructure, and subsequently the MVNO cannot lease any resources from the MNO. We  characterize the unique SPNE of this case. We  compare the results of our model to the benchmark case to assess the incentive of the MNO to invest in her infrastructure and to offer it to the MVNO. 

Analytic and numerical results reveal that although the number of EUs and the price that the MNO quotes for EUs might be lower than those of the MVNO, the MNO would still generate a higher payoff than the MVNO if the per resource fee that the MVNO pays to the MNO is  sufficiently higher than the marginal cost of investment on the infrastructure.  In addition, results reveal that if the MNO knows that EUs are reluctant joining her or have a high preference for the MVNO, then a better strategy for her would be to invest in her infrastructure, lease a portion (or the whole) of the new resources to the MVNO, and collect the benefit through the fees she charges for the resources leased to the MVNO.

The rest of the paper is organized as follows. In Section~\ref{section:model}, we present the model. In Section~\ref{section:analysis}, we find the SPNE strategies. We present the SPNE outcome in Section~\ref{section:outcome}. In Section~\ref{section:benchmark}, we prove the results of the benchmark case. Finally, in Section~\ref{section:numerical}, we present the numerical results and discuss about the results.

\section{Model}\label{section:model}
We consider one MNO and one MVNO that compete in attracting a pool of undecided EUs.  We assume that the MVNO is active, i.e. has already some resources. Thus, the model is set up such that even when the MVNO does not lease any resources from the MNO, it can still attract some of EUs. 

\subsection*{SPs:} 
We denote the MNO by SP$_{L}$ ($L$ represents Leader, since this SP is the leader of the game), and denote the MVNO by SP$_{F}$ ($F$ represents Follower, since this SP is the follower of this leader/follower game). SP$_L$ owns the infrastructure, invests in her infrastructure to attract EUs, and can lease parts of the new resources to SP$_F$ in exchange of money. 

We denote the fee per resources that SP$_{F}$ pays to SP$_{L}$ by $s$. The utility of SP$_F$ is increasing with the revenue from EUs, which depends on the number of EUs ($n_F$) and the access fee ($p_F$). This utility is decreasing with  respect to the fee she pays to SP$_L$ to reserve resources, i.e. $sI^2_F$, where $I_F$ is the number of resources that SP$_F$ reserves from SP$_L$:
\begin{equation} \label{equ:payoffF}
\pi_F=n_F(p_F-c)-sI^2_F
\end{equation}
, where $c$ is the marginal cost associated with each user. Note that naturally, we expect the cost of investing in infrastructure to be convex, i.e. the cost of investment per resources increases with the amount of resources. For simplicity in analysis, we consider the cost of reserving resources  to be  quadratic with respect to $I_F$.\footnote{The overall intuitions of the model are expected to be the same in the case of considering any convex function of $I_F$}

The utility of SP$_L$ is:
\begin{equation}\label{equ:payoffL}
\pi_L=n_L(p_L-c)+sI^2_F-\gamma I^2_L
\end{equation}
where $n_L$, $p_L$, and $I_L$ are the number of EUs with SP$_L$, the access fee quoted by SP$_L$ for EUs, and the number of resources that SP$_L$ add to her infrastructure. Note that $\gamma$ is the marginal cost of investment. Thus, the utility of SP$_L$ is increasing with the revenue from EUs 
($n_L(p_L-c)$) and the fee received from SP$_F$ ($sI^2_F$), and  is decreasing with  respect to the investment cost ($\gamma I^2_L$). Note that we consider the cost of investment to be a quadratic function of $I_L$.\footnote{Intuitions of the model are expected to be the same when considering any convex function of $I_L$.}

Trivially, we assume that the fee per resources ($s$) is greater than or equal to the marginal cost of investment on the infrastructure ($\gamma$), i.e. $s\geq \gamma$. Also, note that $I_F\leq I_L$. To have a non-trivial problem, we also assume $I_L>0$ and $I_F\geq 0$. 

The strategies of SP$_L$ is to choose the access fee for EUs ($p_L$) and the level of investment ($I_L$). The strategies of SP$_F$ is to choose the access fee for EUs ($p_F$) and the number of resources she leases from SP$_L$ ($I_F$). We assume that the per resource fee $s$ is pre determined, possibly through a bargaining framework between SP$_L$ and SP$_F$. 

\subsection*{EUs:}  The strategy of an EU is to choose one of the SPs to buy wireless plan from. We model the EUs using a hotelling model. We assume that the SP$_L$ is located at 0, SP$_F$ is located at 1, and EUs are distributed uniformly along the unit interval $[0,1]$. The closer an EU to an SP, the more this EU prefers this SP to the other. Note that the notion of closeness and distance is used to model the preference of EUs, and may not be the same as  physical distance. 

Specifically, the EU located at $x\in[0,1]$ incurs a \emph{transport cost} of $t_L x$  (respectively, $t_F(1-x)$) when joining SP$_L$ (respectively, SP$_F$), where $t_L$ (respectively, $t_F$) is the marginal transport cost for SP$_L$ (respectively, SP$_F$). In sum, we  consider $t_L$ and $t_F$ as the reluctance of EUs for connecting to SP$_L$ and SP$_F$, respectively. 


We  assume the transport costs of EUs to depend on the investment level of  SPs. Thus, we assume that the higher the investment level of an SP in comparison to the other SP, the higher would be the reluctance of EUs for joining the other SP. This model is used to captures the factor of quality in the decision of EUs. Thus, the transport cost of  SP$_L$ and SP$_F$ are considered to be $t_L=\frac{I_F}{I_L}$ and $t_F=1-t_L=\frac{I_L-I_F}{I_L}$, respectively.

Therefore, the utility of an EU located at distance $x_j$ of  SP$_{j}$, $j\in\{L,F\}$ is:
\begin{equation}
u_{x_j}=v^*-t_jx_j-p_j
\end{equation}
, where $v^*$ is a common valuation that captures the  value of buying a wireless plan for EUs regardless of the SP chosen.

\subsection*{Formulation:}
We model the interaction between SPs and EUs using a four-stage sequential game. Naturally, we assume that SP$_L$ makes the first move and is the leader of the game. The timing and the stages of the game are as follows:
\begin{enumerate}
\item SP$_{L}$ decides on the investment on the infrastructure ($I_L$).
\item SP$_{F}$ decides on the investment, i.e. number of resources to lease from SP$_L$ ($I_F$ such that $I_F\leq I_L$).
\item SP$_{L}$ and SP$_{F}$ determine the access fees for EUs (respectively, $p_L$ and $p_F$).
\item EUs decide to subscribe to one of the SPs.
\end{enumerate}

We assumed the selection of investments by SPs ($I_L$ and $I_F$) to happen before the selection of prices for EUs ($p_L$ and $p_F$) since investment decisions are long-term decision. These decisions are expected to be kept constant over longer time horizons in comparison to pricing decisions. 

In the sequential game framework,  we seek a \emph{Subgame Perfect Nash Equilibrium}  using \emph{backward induction}:

\begin{definition}\emph{Subgame Perfect Nash Equilibrium (SPNE):}
	A strategy is an SPNE if and only if it constitutes a Nash Equilibrium (NE) of every subgame of the game. 
\end{definition} 

\begin{definition}\emph{Backward Induction:} Characterizing the equilibrium strategies starting from the last stage of the game and proceeding backward. 
\end{definition}

We also assume the full market coverage of EUs by SPs. This means that  each EU chooses exactly one SP to subscribe to.  This  assumption is common in hotelling models and is necessary to ensure competition between SPs. An equivalent assumption is to consider the common valuation $v^*$ to be sufficiently large so that the utility of EUs for buying a wireless plan  is positive regardless of the choice of SP.

 One might think that SP$_L$ can eliminate competition by not offering her resources to SP$_F$ (or simply not investing in  new infrastructure), and use the full market coverage to act as a monopoly. However, this is not the case in our model. Recall that we assumed the SP$_F$ to have resources independent of SP$_L$. Thus, regardless of the strategy of SP$_L$, SP$_F$ would be still active in the market, and the full market coverage would not lead to the monopoly of SP$_L$. 
 


\section{The Sub-Game Perfect Nash Equilibrium}\label{section:analysis}

In the backward induction, we start with Stage 4:

\subsection*{Stage 4:} 

In this subsection, we characterize the division of EUs between SPs in the  equilibrium, i.e. $n_L$ and $n_F$, using the knowledge of the strategies chosen by the  SPs in Stages 1, 2, and 3. To do so, we characterize the location of the EU that is indifferent between joining either of the SPs, $x_n$. Thus, EUs located at $[0,x_n)$ join SP$_L$, and those located at $(x_n,1]$ joins the non-neutral SP$_F$ (using the full market coverage assumption). 

The EU located at $x_n\in[0,1]$ is indifferent between connecting to SP$_L$ ans SP$_F$ if:
\begin{equation}\label{equ:xn_t_parameter}
\small
	\begin{aligned}
	&v^*-t_{F}(1-x_n)-p_{F}=v^*-t_Lx_n-p_L\\
	&\qquad \qquad \qquad  \Rightarrow x_n=\frac{t_F+p_F-p_L}{t_L+t_F}
	\end{aligned}
\end{equation}

Note that $t_F+t_L=1$. Substituting the value of $t_F$ yields:
\begin{equation}\label{equ:xn}
\small
x_n=\frac{I_L-I_F}{I_L}+p_F-p_L
\end{equation}

Thus, the fraction of EUs with each SP ($n_L$ and $n_F$) is:
\begin{equation} \label{equ:EUs_stage4}
\small
\begin{aligned} 
n_L &=
\left\{
\begin{array}{ll}
	0  & \mbox{if } x_n < 0 \\
	\frac{I_L-I_F}{I_L}+p_F-p_L & \mbox{if } 0\leq x_n \leq 1\\
	1 & \mbox{if } x_n>1
\end{array}
\right. \\
n_F&=1-n_L
\end{aligned}
\end{equation}

\subsection*{Stage 3:}

In this stage, SP$_L$ and SP$_F$ determine their  prices for EUs, $p_L$ and $p_F$, respectively, to maximize their payoff. We seek for Nash equilibrium (NE) strategies. Note that in general there might exist several NE strategies: some of them corner equilibria (an extreme case in which one of the SPs receives zero EUs) and some interior equilibria (in which both SPs receive a positive mass of EUs). In practice, the latter equilibria are expected to occur more frequently. Thus, henceforth,  we seek to characterize the interior  equilibria, i.e. when $0<n_L<1$ and $0<n_F<1$. 

Thus, we look for all NE by which  $0< x_n<1$ ($x_n$ characterized in the previous stage of the game). In this case, using \eqref{equ:payoffF}, \eqref{equ:payoffL}, and \eqref{equ:EUs_stage4}, the payoffs of SPs would be:
\begin{equation}\label{equ:payoffF_2}
\small
\pi_F=(\frac{I_F}{I_L}+p_L-p_F)(p_F-c)-sI^2_F 
\end{equation}
\begin{equation}\label{equ:payoffL_2}
\small
\pi_L=(\frac{I_L-I_F}{I_L}+p_F-p_L)(p_L-c)+sI^2_F-\gamma I^2_L
\end{equation}
\normalsize2
In the following theorem, we prove that the NE uniquely exists and  characterize it:
\begin{theorem}\label{theorem:stage 3}
The NE strategies of Stage 3 by which  $0<n_L<1$ and $0<n_F<1$ are unique, and are $p_F=c+\frac{I_L+I_F}{3I_L}$ and $p_L=c+\frac{2I_L-I_F}{3I_L}$. 
\end{theorem}

\begin{remark}
Note that $\frac{d p_L}{d I_L}\geq 0$ and $\frac{d p_F}{d I_L}\leq 0$. Thus, as we expect intuitively, $p_L$ (respectively, $p_F$) is increasing (respectively, decreasing) with respect to $I_L$. Also, $p_L$ (respectively, $p_F$) is decreasing (respectively, increasing) with respect to $I_F$. 
\end{remark}

\begin{proof}
In this case, every NE should satisfy the first order condition. Thus, $p^*_F$ and $p^*_L$ should be determined such that $\frac{d \pi_F}{dp_F}=0$ and   $\frac{d \pi_L}{dp_L}=0$.  The first order conditions yield:
$$
\small
\begin{aligned} 
2p^*_F-p^*_L&=\frac{I_F}{I_L}+c \qquad \& \qquad 2p^*_L-p^*_F&=\frac{I_L-I_F}{I_L}+c 
\end{aligned}
$$
Thus,
\begin{eqnarray}\label{equ:optimump}
\small
\begin{aligned}
p^*_F=c+\frac{I_L+I_F}{3I_L} \qquad \& \qquad p^*_L=c+\frac{2I_L-I_F}{3I_L} 
\end{aligned}
\end{eqnarray}
\normalsize
Therefore, $p^*_F$ and $p^*_L$ are the unique NE strategies if they yield $0< x_n< 1$ and no unilateral deviation is profitable for SPs. In Case A, we check the former, and in Case B, we check the latter. 

\textbf{Case A:}
We first  check the condition that with $p^*_L$ and $p^*_F$, $0< x_n< 1$. Using \eqref{equ:xn} and \eqref{equ:optimump}:
$$
\small
x_n=\frac{I_L-I_F}{I_L}+p^*_F-p^*_L= \frac{2I_L-I_F}{3I_L}
$$
, which is greater than zero since $I_L\geq I_F$ and $I_L>0$. $x_n$ is also clearly less than one (note that $I_F\geq 0$). Thus, the condition $0< x_n< 1$ holds.

\textbf{Case B:}
Note that $\frac{d^2 \pi_F}{dp^2_F}<0$ and   $\frac{d^2 \pi_L}{dp^2_L}<0$. Thus, any solutions to the first order conditions would maximize the payoff of the SPs when $0<x_n<1$, and no unilateral deviation by which $0<x_n<1$ would be profitable for SPs. 

Now, we discuss that any deviation by SPs such that $n_L=0$ and $n_L=1$ (which subsequently yields $n_F=1$ and $n_F=0$, respectively) is not profitable. Note that the payoff of SPs, \eqref{equ:payoffF} and \eqref{equ:payoffL}, are continuous  as $n_L\downarrow 0$, and $n_L\uparrow 1$ (which subsequently yields  $n_F\uparrow 1$ and $n_F\downarrow 0$, respectively). Thus, the payoffs of both SPs when selecting $p^*_L$ and $p^*_F$ (solutions of first order conditions) are greater than or equal to the payoffs when $n_L=0$ and $n_L=1$. Thus, any deviation by SPs such that $n_L=0$ or $n_L=1$ is not profitable for SPs. 


This complete the proof that $p^*_F$ and $p^*_L$ are the unique NE strategies by which both SPs are active, i.e. $0<x_n<1$.  
\end{proof}
 
\subsection*{Stage 2:}

In this stage of the game, SP$_F$ decides on the investment level, i.e. the number of resources to be leased from SP$_L$ ($I_F$), with the condition that $I_F\leq I_L$ to maximize $\pi_F$:
\begin{equation} \label{equ:maxpiFIF}
\max_{I_F\leq I_L}\pi_F=\max_{I_F\leq I_L}\big{(}\frac{I_L+I_F}{3I_L}\big{)}^2-sI^2_F
\end{equation} 
Note that for the last equality, we used  \eqref{equ:payoffF_2} and Theorem~\ref{theorem:stage 3}. 

\begin{theorem}
The optimum investment level of SP$_F$ is:
\begin{equation}\label{equ:optimumI_F}
I^*_F=\left\{
		\begin{array}{ll}
			\frac{I_L}{9sI^2_L-1} & \mbox{if } \quad I_L>\sqrt{\frac{2}{9s}}\\
			 I_L & \mbox{if } \quad  I_L\leq \sqrt{\frac{2}{9s}}
		\end{array}
	\right.
	\end{equation}
\end{theorem}

\begin{remark}
If the fee per resources ($s$) or the investment by SP$_L$ ($I_L$) is low, then SP$_F$ reserves all the available resources. If not, then SP$_F$ reserves a fraction of available resources ($I^*_F<I_L$). Note that in this case, $\frac{d I^*_F}{d I_L}<0$. Thus, the higher the number of available resources, the lower would be the number of resources reserved by SP$_F$. 
\end{remark}

\begin{proof}
Note that in \eqref{equ:maxpiFIF}, $\pi_F$ is concave over $I_F$ if $s>\frac{1}{9I^2_L}$, $\pi_F$ is convex if $s<\frac{1}{9I^2_L}$, and is linear if $s=\frac{1}{9I^2_L}$.  We characterize the optimum investment level, i.e. $I^*_F$ under each of these conditions:

\subsection*{$s>\frac{1}{9I^2_L}$:}  In this case, the payoff of SP$_F$ is concave. Thus, the first order condition yields the possible optimum investment level, $\hat{I}_F$:
\begin{equation} \label{equ:hatI_f}
\begin{aligned}
\small
\frac{d \pi_F}{d I_F}|_{\hat{I}_F}=0&\Rightarrow \frac{2(\hat{I}_F+I_L)}{9I^2_L}-2s\hat{I}_F=0\\
&\Rightarrow \hat{I}_F=\frac{I_L}{9sI^2_L-1}
\end{aligned}
\end{equation}
Thus, the optimum investment level, $I^*_F$ is:
\begin{equation}
I^*_F=\min\{\hat{I}_F,I_L\} 
\end{equation}
Note that from the assumption of this case $9sI^2_L>1$. Thus, $\hat{I}_F>0$ (using \eqref{equ:hatI_f}). In addition, from the expression of $\hat{I}_F$ \eqref{equ:hatI_f}:
$$
\hat{I}_F<I_L\iff 9sI^2_L>2
$$ 
Thus, the optimum investment level for SP$_F$ would be:
\begin{equation}
\small
I^*_F=\left\{
		\begin{array}{ll}
			\frac{I_L}{9sI^2_L-1} & \mbox{if } \quad I_L>\sqrt{\frac{2}{9s}}\\
			 I_L & \mbox{if } \quad  I_L\leq \sqrt{\frac{2}{9s}}
		\end{array}
	\right.
\end{equation}
\subsection*{$s<\frac{1}{9I^2_L}$:} In this case, $\pi_F$ would be convex. Thus, the optimum level of investment would be on the boundaries of the feasible set. Thus, $I^*_F$ would be either 0 or $I_L$, whichever yields a higher payoff. Note that:
$$
\begin{aligned}
\pi_F|_{I_F=0}&= \frac{1}{9}\qquad \& \qquad \pi_F|_{I_F=I_L}&=\frac{4}{9}-sI^2_L
\end{aligned}
$$
Thus, $I^*_F=I_L$ if and only if $\frac{4}{9}-sI^2_L\geq \frac{1}{9}$ (Note that we assumed that if $I_F=0$ and $I_F=I_L$ yield the same payoff, then SP$_F$ chooses $I_F=I_L$). Thus, $s\leq \frac{1}{3I^2_L}$ yields $I^*_F=I_L$. Note that from the assumption of the case,  $s<\frac{1}{9I^2_L}<\frac{1}{3I^2_L}$. Thus, $I^*_F=I_L$.

\subsection*{$s=\frac{1}{9I^2_L}$:} In this case, $\pi_F$ would be an increasing linear function of $I_F$. Thus, $I^*_F=I_L$. 

Putting all the cases together, the result of the theorem follows.
\end{proof}

\subsection*{Stage 1:} In this stage, SP$_L$ decides on the level of investment $I_L$ to maximize her payoff, $\pi_L$: 
\begin{equation}\label{equ:optimizationpil}
\max _{I_L}\pi_L=\max_{I_L} (\frac{2I_L-I^*_F}{3I_L})^2+sI^{*2}_F-\gamma I^2_L
\end{equation}
where for the last equality, we used  \eqref{equ:payoffL_2} and Theorem \ref{theorem:stage 3}, and $I^*_F$ is characterized in \eqref{equ:optimumI_F}.  In the next theorem, we characterize the candidate optimum answers:

\begin{theorem}\label{theorem:stage 1}
The optimum solution to \eqref{equ:optimizationpil} is $I^*_L=\min\{\sqrt{\frac{2}{9s}},\hat{I}_L\}$, where $\hat{I}_L$ is the solution of the first order condition on:
\begin{equation} \label{equ:payoff_1storder}
\pi_{L,I^*_F}=\frac{1}{9}(2-\frac{1}{9sI^2_L-1})^2+s\frac{I_L}{9sI^2_L-1}-\gamma I^2_L
\end{equation}
\end{theorem}

\begin{remark}
Theorem yields that the minimum optimum level of investment by SP$_L$ is $\sqrt{\frac{2}{9s}}$. 
\end{remark}

\begin{remark}\label{remark:optimumIL}
Note that the first term in \eqref{equ:payoff_1storder} is increasing with $I_L$ (since the number of EUs is increasing with $I_L$). The second term is the payment from SP$_F$ which is decreasing with respect to $I_L$ (since when $I_L>\sqrt{\frac{2}{9s}}$, $I^*_F$ is decreasing with respect to $I_L$). The third term is the cost of investment which is decreasing with $I_L$. Thus, when either $s$ or $\gamma$ is sufficiently large, we expect the utility \eqref{equ:payoff_1storder} to be decreasing with respect to $I_L$, and the optimum answer to be $I^*_L=\sqrt{\frac{2}{9s}}$. On the other hand, $\hat{I}_L$ is expected to be the optimum answer when both $s$ and $\gamma$ are sufficiently small. We discuss about this in Numerical Results (Section~\ref{section:numerical}). Numerical results reveal that for large sets of parameters, we expect $I^*_F=\sqrt{\frac{2}{9s}}$. 
\end{remark}

\begin{proof}
To find the optimum level of investment, $I^*_L$, consider two scenarios:
\subsection*{$I_L\leq \sqrt{\frac{2}{9s}}$:} In this case, $I^*_F=I_L$ by \eqref{equ:optimumI_F}. Thus, the maximization \eqref{equ:optimizationpil} would become:
$$
\max_{I_L\leq \sqrt{\frac{2}{9s}}}(s-\gamma) I^2_L
$$
Thus:
\begin{equation}
I^*_L=\left\{
		\begin{array}{ll}
			\sqrt{\frac{2}{9s}} & \mbox{if } \quad s\geq \gamma\\
			 0 & \mbox{if } \quad  s< \gamma
		\end{array}
	\right.
\end{equation}
Note that we assumed that if $I_L>0$ and $I_L=0$ yield the same payoff for  SP$_L$, then this SP chooses $I_L>0$. This is the reason that for $s=\gamma$, $I_L>0$ is chosen. Also, recall that to have a non-trivial problem, we assumed $s\geq \gamma$. Thus, if $I_L\leq \sqrt{\frac{2}{9s}}$, then $I^*_L=\sqrt{\frac{2}{9s}}$. 

\subsection*{$I_L>\sqrt{\frac{2}{9s}}$:} In this case, $I^*_F=\frac{I_L}{9sI^2_L-1}$. Thus, the maximization \eqref{equ:optimizationpil} would become:
\begin{equation}
\small
\max_{I_L>\sqrt{\frac{2}{9s}}}{\pi_{L,I^*_F}=\frac{1}{9}(2-\frac{1}{9sI^2_L-1})^2+s\frac{I_L}{9sI^2_L-1}-\gamma I^2_L}
\end{equation}

Note that $I^*_F$ \eqref{equ:optimumI_F} is continuous at $I_L=\sqrt{\frac{2}{9s}}$. Thus, as $I_L\rightarrow \sqrt{\frac{2}{9s}}$, $\pi_{L,I^*_F}\rightarrow \pi_L|_{I_L=\sqrt{\frac{2}{9s}}}$, which is considered in the previous case.  Thus, the only possible optimum answer is the solution to the first order condition on $\pi_{L,I^*_F}$, i.e. $\hat{I}_L$ if $\hat{I}_L>\sqrt{\frac{2}{9s}}$. Results of the theorem follows.

\end{proof}

\section{The Outcome of the Game} \label{section:outcome}

In this section, we characterize the equilibrium outcome of the game using the results of the previous section and discuss about them. In Corollaries~\ref{corollary:outcome_IL>} and \ref{corollary:outcome_IL<}, we characterize
the equilibrium outcome when $I^*_L>\sqrt{\frac{2}{9s}}$  and $I^*_L\leq \sqrt{\frac{2}{9s}}$, respectively. Note that in each case, there exists a unique outcome and the two cases are mutually exclusive. Thus, the outcome of the game exists and is unique. 

\subsection{If $I^*_L>\sqrt{\frac{2}{9s}}$:}
\begin{corollary} [Outcome A]\label{corollary:outcome_IL>}
If $I^*_L>\sqrt{\frac{2}{9s}}$, then:
\begin{itemize}
\item\textbf{Stage 1:}  The optimum level of investment of SP$_L$ is $I^*_L=\hat{I}_L$ (characterized in Theorem~\ref{theorem:stage 1}).
\item \textbf{Stage 2:} The optimum level of investment by SP$_F$ is: $I^*_F=\frac{I^*_L}{9sI^{* 2}_L-1}$.
\item \textbf{Stage 3:} Prices for EUs are: $p^*_F=c+\frac{1+\frac{1}{9sI^2_L-1}}{3}$ and $p^*_L=c+\frac{2-\frac{1}{9sI^2_L-1}}{3}$. 
\item \textbf{Stage 4:} The fractions of EUs with each SP are: $n^*_F=\frac{1+\frac{1}{9sI^2_L-1}}{3}$ and $n^*_L=\frac{2-\frac{1}{9sI^2_L-1}}{3}$.
\end{itemize}
\end{corollary}
\begin{proof}
The outcome of Stage 2 directly follows from the results of Stage 2. The outcome of Stage 3 follows from \eqref{equ:optimump} and the outcome of Stage 2. The outcome of Stage 4 follows from  \eqref{equ:EUs_stage4} and the outcomes of Stage 2 and 3. 
\end{proof}

Note that in this case, as discussed, $I^*_F$ is decreasing with respect to $I^*_L$ and $s$. In addition, given that $I_L$ is fixed, $p^*_F$ and $p^*_L$ are decreasing and increasing with respect to the per resource fee, $s$, respectively. Also, given that $s$ is fixed, $p^*_F$ and $p^*_L$ are decreasing and increasing with $I^*_L$, respectively.  Regardless of $p^*_F$ (respectively, $p^*_L$) being decreasing (respectively, increasing), the number of EUs with SP$_F$ (respectively, SP$_L$), i.e. $n^*_F$ (respectively, $n^*_L$) is still decreasing (respectively, increasing) with respect to $I^*_L$ (if $s$ fixed). This is because of the increase (respectively, decrease) in $t_F$ (respectively, $t_L$), i.e. the transport cost of SP$_F$ (respectively, SP$_L$). To understand these changes in the transport costs, recall that increasing $I^*_L$, decreases $I^*_F$, and subsequently increases $t_F$.

 In  Section~\ref{section:numerical}, we  observe that $I^*_L$ is dependent on $s$. Thus, the relationship between $s$ and the outcome is more complicated.

\subsection{If $I^*_L\leq \sqrt{\frac{2}{9s}}$:}

\begin{corollary}[Outcome B]\label{corollary:outcome_IL<}
If $I^*_L\leq \sqrt{\frac{2}{9s}}$, then:
\begin{itemize}
\item \textbf{Stage 1:} $I^*_L= \sqrt{\frac{2}{9s}}$.
\item \textbf{Stage 2:} $I^*_F=I^*_L$.
\item \textbf{Stage 3:} $p^*_F=c+\frac{2}{3}$ and $p^*_L=c+\frac{1}{3}$.
\item \textbf{Stage 4:}  $n^*_F=\frac{2}{3}$ and $n^*_L=\frac{1}{3}$.
\end{itemize}
\end{corollary}
Proof is similar to the proof of the previous corollary. 

In this case,  SP$_F$ reserves all available resources, and  the investment level of SP$_L$ (which is equal to the number of resources reserved by SP$_F$) is a decreasing function of the fee per resources, i.e. $s$. SP$_F$ quotes a higher price for EUs in comparison to SP$_L$. In spite of the higher price, SP$_F$ would be able to attract more EUs given the better investment level in comparison to SP$_L$ which translates into a lower transport cost. 

We  calculate $\pi_L$ and $\pi_F$ using Corollary~\ref{corollary:outcome_IL<} and \eqref{equ:payoffL}:
\begin{equation}
\begin{aligned}
\pi^*_L&=\frac{1}{3}-\frac{2 \gamma}{9 s}\qquad \& \qquad \pi^*_F&=\frac{2}{9}
\end{aligned}
\end{equation}
Thus, the payoff of SP$_L$ would be higher than the payoff of SP$_F$, i.e. $\pi^*_L>\pi^*_F$, if and only if $s>2\gamma$. In this case, although in comparison to SP$_F$, SP$_L$ offers her service to EUs with lower price and  attracts a lower number of EUs, she can still obtain a higher payoff through the per resource fee that she collects from SP$_F$. Thus, SP$_L$ leases the resources to SP$_F$ and instead generates revenue through the fee she charges to SP$_F$. 

\section{Benchmark Case} \label{section:benchmark}

Now, we consider the case in which SP$_L$ does not invest in her infrastructure. Similar to the original model, both SP$_L$ and SP$_F$ compete to attract a pool of undecided users. The hotelling model for EUs is the same as before. The only difference is the transport cost. Here, we consider the transport costs to be parameters, and denote them by $t_L>0$ and $t_F>0$ for SP$_L$ and SP$_F$, respectively. We discuss about the effects of these parameters in the Numerical Results (Section~\ref{section:numerical}).  

The analysis of stage 4 of the game would be the same as before, and by \eqref{equ:xn_t_parameter}:

\begin{equation}\label{equ:xn_banchmark}
\small
x_n=\frac{t_F+p_F-p_L}{t_L+t_F} 
\end{equation}
Thus,
\begin{equation} \label{equ:EUs_stage4_benchmark}
\small
\begin{aligned} 
n_L &=
\left\{
\begin{array}{ll}
	0  & \mbox{if } x_n < 0 \\
	\frac{t_F+p_F-p_L}{t_L+t_F} & \mbox{if } 0\leq x_n \leq 1\\
	1 & \mbox{if } x_n>1
\end{array}
\right. \\
n_F&=1-n_L
\end{aligned}
\end{equation}  

Note that  in the benchmark case, there is no investment decision. Thus, the payoffs of SPs would be modified as:

\begin{equation} \label{equ:payoff_F_bench}
\small
\pi_F=n_F(p_F-c)
\end{equation}
\begin{equation} \label{equ:payoff_L_bench}
\small
\pi_L=n_L(p_L-c)
\end{equation}

Thus, the only decision of SPs is to determine their pricing for EUs. We use the same approach as Stage 3 of the game to characterize the NE pricing strategies for SPs by which $0<x_n<1$:

\begin{theorem}
The NE pricing strategies of SPs are:
\begin{eqnarray}
\small
\begin{aligned}
p^*_F&=c+\frac{2t_L+t_F}{3}\\
p^*_L&=c+\frac{t_L+2t_F}{3} 
\nonumber
\end{aligned}
\end{eqnarray}
\end{theorem}
\begin{proof}
Proof is similar to the proof of Theorem~\ref{theorem:stage 1}.  $p^*_F$ and $p^*_L$ should be determined such that $\frac{d \pi_F}{dp_F}=0$ and   $\frac{d \pi_L}{dp_L}=0$:
\begin{eqnarray}\label{equ:optimump_benchmark}
\small
\begin{aligned}
p^*_F&=c+\frac{2t_L+t_F}{3}\qquad \& \qquad p^*_L&=c+\frac{2t_F+t_L}{3} 
\end{aligned}
\end{eqnarray}

Therefore, $p^*_F$ and $p^*_L$ are the unique NE strategies if they yield $0< x_n< 1$ and no unilateral deviation is profitable for SPs. In Case A, we check the former, and in Case B, we check the latter. 

\textbf{Case A:}
We first  check the condition that with $p^*_L$ and $p^*_F$, $0< x_n< 1$. Using \eqref{equ:xn_banchmark} and \eqref{equ:optimump_benchmark}:
$$
x_n=\frac{2t_F+t_L}{3(t_F+t_L)}
$$
, which is greater than zero since $t_F>0$ and $t_L>0$. $x_n$ is also clearly less than one. Thus, the condition $0< x_n< 1$ holds.

\textbf{Case B:}
Note that $\frac{d^2 \pi_F}{dp^2_F}<0$ and   $\frac{d^2 \pi_L}{dp^2_L}<0$. Thus, any solutions to the first order conditions would maximize the payoff of the SPs when $0<x_n<1$, and no unilateral deviation by which $0<x_n<1$ would be profitable for SPs. 

Now, we discuss that any deviation by SPs such that $n_L=0$ and $n_L=1$ (which subsequently yields $n_F=1$ and $n_F=0$, respectively) is not profitable. Note that the payoff of SPs, \eqref{equ:payoff_F_bench} and \eqref{equ:payoff_L_bench}, are continuous  as $n_L\downarrow 0$, and $n_L\uparrow 1$ (which subsequently yields  $n_F\uparrow 1$ and $n_F\downarrow 0$, respectively). Thus, the payoffs of both SPs when selecting $p^*_L$ and $p^*_F$ (solutions of first order conditions) are greater than or equal to the payoffs when $n_L=0$ and $n_L=1$. Thus, any deviation by SPs such that $n_L=0$ or $n_L=1$ is not profitable for SPs. 


This complete the proof that $p^*_F$ and $p^*_L$ are the unique NE strategies by which both SPs are active, i.e. $0<x_n<1$.  
\end{proof}

Note that in the absence of  investments, Stage 2 and 1 would be of no importance. In the following corollary, we characterize the outcome of the game in the benchmark case (subscript B stands for Benchmark):
\begin{corollary}
The equilibrium outcome of the benchmark case is as follows:
\begin{itemize}
\item NE pricing strategies are: $p^*_{F,B}=c+\frac{2t_L+t_F}{3}$ and $p^*_{L,B}=c+\frac{t_L+2t_F}{3}$.
\item Fractions of EUs with each SP: $n^*_{L,B}=\frac{2t_F+t_L}{3(t_F+t_L)}$ and $n^*_{F,B}=\frac{t_F+2t_L}{3(t_F+t_L)}$.
\end{itemize}
\end{corollary}

Using \eqref{equ:payoff_F_bench} and \eqref{equ:payoff_L_bench}, the payoffs of SPs in these cases are:
\begin{equation}
\begin{aligned}
\pi^*_{F,B}&=(\frac{t_F+2t_L}{3(t_F+t_L)})^2\qquad \& \qquad \pi^*_{L,B}&=(\frac{2t_F+t_L}{3(t_F+t_L)})^2
\end{aligned}
\end{equation}

Note that if $t_L>t_F$, then $p^*_{F,B}>p^*_{L,B}$, $n^*_{F,B}>n^*_{L,B}$, and subsequently $\pi^*_{F,B}>\pi^*_{L,B}$, and vice versa. Thus, the SP that EUs have lower reluctance for, receives the highest payoff.

\section{Numerical Results  and Discussions}  \label{section:numerical}

In this section, we use numerical simulations (i) to determine whether and under what conditions the outcome in Corollary~\ref{corollary:outcome_IL>} (Outcome A) would emerge, and (ii)  to provide insights for results under different parameters of the model. For all results, we consider $c=1$.\footnote{Note that the choice of $c$ barely affects the results. It may only shift some of the results (e.g. the price charged to EUs) by only a constant.}

\begin{figure}[t]
	\begin{subfigure}{.25\textwidth}
		\centering
		\includegraphics[width=\linewidth]{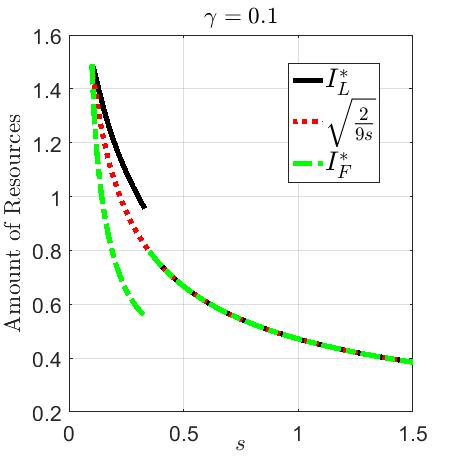}
		\label{figure:resources_gamma0.1}
	\end{subfigure}%
	\begin{subfigure}{.25\textwidth}
		\centering
		\includegraphics[width=\linewidth]{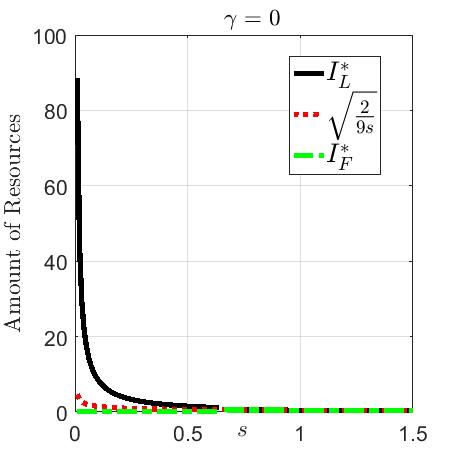}
		\label{fig:resources_gamma0}
	\end{subfigure}
	\caption{Investment decisions of SPs vs. per resource fee, $s$}\label{figure:resources}
\end{figure}

In Figure~\ref{figure:resources}, we plot the optimum level of investment of SP$_L$ ($I^*_L$), the number of resources that SP$_F$ reserves ($I^*_F$), and the minimum optimum level of investment by SP$_L$ ($\sqrt{\frac{2}{9s}}$), when $\gamma=0.1$ (left) and $\gamma=0$ (right) (recall that,  in \eqref{equ:payoffL}, $\gamma$ is the marginal cost of investment.). The discontinuities  are because of the transition of the outcome of the game from the Outcome A to Outcome B (outcome in Corollary~\ref{corollary:outcome_IL<}). Results reveal that for each outcome, $I^*_L$ and $I^*_F$ are decreasing with the per resource fee, $s$. Thus, the higher the fee per resources that SP$_F$ pays to SP$_L$, the lower would be the number of resources that SP$_F$ reserves, and subsequently the lower  would be the investment of SP$_L$.  Also, when the marginal cost of investment ($\gamma$) is zero, SP$_L$ investment is significantly more in comparison to $\gamma=0.1$. Thus, intuitively, the higher the investment cost, the lower would be the investment.


Results also confirm the intuitions presented in Remark~\ref{remark:optimumIL} that for sufficiently small $s$ and $\gamma$, the optimum level of investment by SP$_L$ would be equal to $\hat{I}_L$ (introduced in Theorem~\ref{theorem:stage 1}) which is  higher than $\sqrt{\frac{2}{9s}}$ (Outcome A). Numerical results also reveal that for $\gamma>0.12$, Outcome A would not occur even for small $s$. Thus, the outcome in which $\hat{I}_L$ is the optimum level of investment only occurs for a small set of parameters. Therefore, for a wide range of parameters, we expect Outcome B to be the outcome of the game.

In Figure~\ref{figure:prices}, we plot the pricing decisions of SPs for EUs, i.e. $p^*_L$ and $p^*_F$, for $\gamma=0.1$ (left) and $\gamma=0$ (right). Similar to Figure~\ref{figure:resources},  the discontinuities in the figures are because of the transition of the outcome of the game from Outcome A to B.  Note that in Outcome B (when $s$ is sufficiently large), $p^*_F$ and $p^*_L$ are constant (given that c is a constant) independent of $\gamma$ and $s$, and the price that SP$_F$ charges is twice the price that SP$_L$ charges. However, in Region A, if $\gamma$, i.e. the marginal cost of investment, is extremely small (e.g. zero), then SP$_L$ would be able to charge a higher price than SP$_F$ (Figure~\ref{figure:prices}-right). The reason is that for $\gamma$ small, SP$_L$ will invest more ($I^*_L$ small). We also stated in Section~\ref{section:outcome} after Corollary~\ref{corollary:outcome_IL>} that  $p^*_F$  and $p^*_L$ are decreasing and increasing with $I^*_L$, respectively. Thus, small investment cost yields a higher investment by SP$_L$ and as a result a higher price for EUs of this SP. 

\begin{figure}[t]
	\begin{subfigure}{.25\textwidth}
		\centering
		\includegraphics[width=\linewidth]{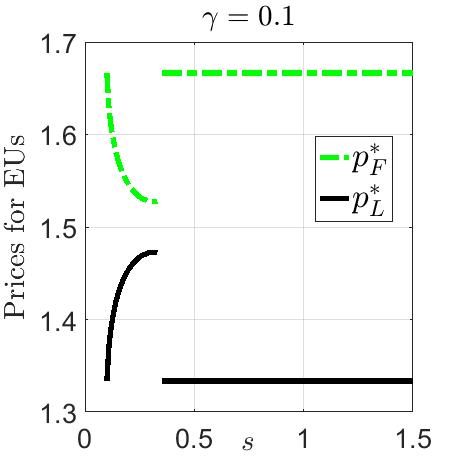}
		\label{fig:prices_gamma01}
	\end{subfigure}%
	\begin{subfigure}{.25\textwidth}
		\centering
		\includegraphics[width=\linewidth]{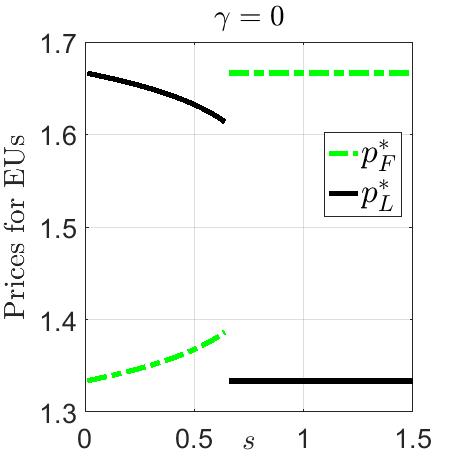}
		\label{fig:prices_gamma0}
	\end{subfigure}
	\caption{Pricing decisions of SPs for EUs versus $s$}\label{figure:prices}
\end{figure}

Also note that  in Outcome A,  depending on $\gamma$, prices can be increasing or decreasing with $s$: If the marginal cost of investment is extremely small, then $p^*_F$ is increasing with $s$ and $p^*_L$ is decreasing with $s$ (Figure~\ref{figure:prices}-right). The opposite is true when $\gamma$ is not small  (Figure~\ref{figure:prices}-left). To describe the reason behind this result, note that from Corollary~\ref{corollary:outcome_IL>}, $p^*_L$ and $p^*_F$ are increasing and decreasing, respectively, with respect to both $s$ and $I^*_L$. On the other hand, $I^*_L$ is itself decreasing with $s$ (Figure~\ref{figure:resources}). Thus, as $s$ increases two factors one increasing (s) and one decreasing ($I^*_L$) affect the prices. Numerical results yield that when $\gamma$ is extremely small (thus, $I^*_L$ is large), the rate of change in $I^*_L$ with $s$ dominates the rate of change in $s$. Thus, as $s$ increases, $p^*_F$ increases and $p^*_L$ decreases. 

Note that by Corollaries \ref{corollary:outcome_IL>} and \ref{corollary:outcome_IL<}, the dependency of $n^*_L$ and $n^*_F$ to parameters of the model is  similar to the dependency of $p^*_L$ and $p^*_F$. The only difference is the exclusion of $c$ from the expressions.
 
In Figure~\ref{figure:payoffs},  we plot the payoffs of SPs, i.e. $\pi^*_F$ and $\pi^*_L$, and the payoff of SP$_L$ in the benchmark case for three scenarios (i) $t_F=0.5$ and $t_L=0.5$, (ii) $t_F=1$ and $t_L=0$, and (iii) $t_F=0$ and $t_L=1$. Results reveal that although for some parameters, the price that SP$_F$ quotes for EUs and the number of EUs that she can attract is higher than the price  of SP$_L$  and the number of EUs of this SP, SP$_L$ can still fetch a higher payoff than SP$_F$. The reason is the fee that SP$_F$ pays to SP$_L$.   
 
\begin{figure}[t]
	\centering
	\includegraphics[width=0.4\textwidth]{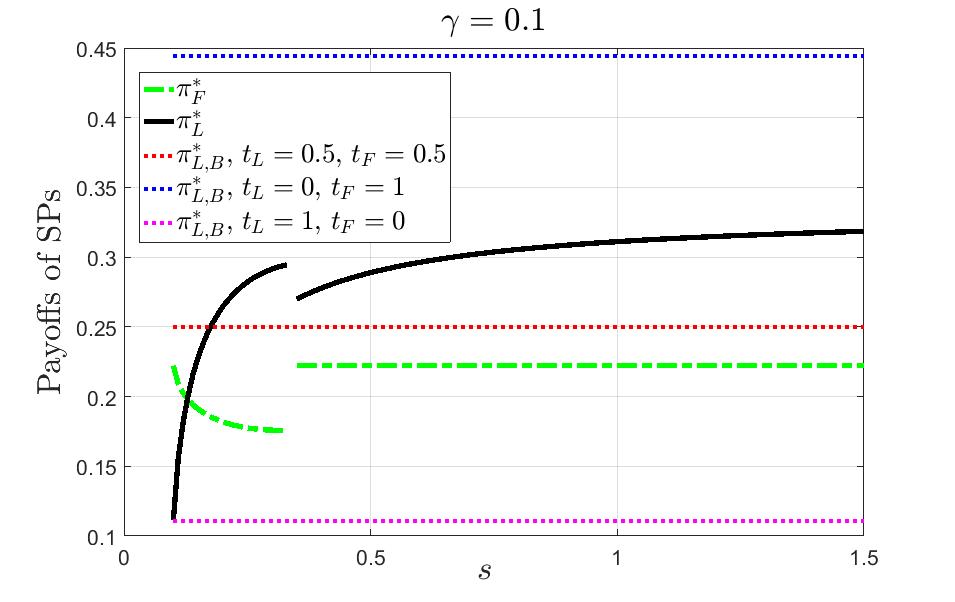}
	\caption{Payoffs of SPs for EUs versus $s$}
	\label{figure:payoffs}
\end{figure}
 
In addition, we  can only expect an investment by SP$_L$ if this investment yields a higher payoff than the benchmark case with no investment. Results in Figure~\ref{figure:payoffs} yield that the higher the transport cost of SP$_L$ in comparison to SP$_F$, the lower would be the payoff of SP$_L$ in the benchmark case, and the higher would be the incentive of this SP to invest and lease her resources to SP$_F$. Thus, the preferences of EUs for SPs, i.e. the transport costs, are one of the main factors in gauging the plausibility of a cooperation between an MNO and an MVNO. If an MNO knows that EUs are reluctant joining her, i.e. EUs have a high transport costs for the MNO, then she may invest in her infrastructure, lease a portion of the new resources to an independent MVNO, which EUs have lower reluctance for, and collect the profit through the fees she charges to the MVNO.

In Figure~\ref{figure:payoffs}, note that if $s$ is small, then SP$_F$ fetches a higher payoff than SP$_L$ since she can get a higher revenue from EUs (Figure~\ref{figure:prices}-left and recalling that $n^*_F$ has the same pattern as $p^*_F$). Overall, increasing $s$ increases the payoff of SP$_L$. Counter intuitively, when $s$ is large enough such that Outcome B occurs (Corollary~\ref{corollary:outcome_IL<}) SP$_F$ receives a higher payoff in comparison to her payoffs when $s$ is small. This occurs since in this case, SP$_F$ reserves all the available resources from SP$_L$. This  enables SP$_F$ to charge a high price for EUs while still attracting EUs. Thus, when $s$ is large enough, both SPs receive high revenue at the expense of EUs paying more. Thus, the welfare of EUs would be low. This highlights the importance  of a framework by which the per resource fee is determined such that all entities of the market can benefit. This framework could be a bargaining game between SPs with some restrictions imposed by a regulator. This is a topic of future work. 

\bibliographystyle{IEEEtran}
\bibliography{bmc_article}

\end{document}